\documentclass{elsarticle}
\usepackage[margin=3cm]{geometry}

\makeatletter
\def\ps@pprintTitle{%
	\let\@oddhead\@empty
	\let\@evenhead\@empty
	\let\@oddfoot\@empty
	\let\@evenfoot\@oddfoot
}
\makeatother

\usepackage[noend, boxruled, linesnumbered]{algorithm2e}
\usepackage{amssymb}

\usepackage[english]{babel}
\usepackage{natbib}

\usepackage{multicol}
\usepackage{url}
\usepackage{amssymb}
\usepackage{amsfonts}
\usepackage{amsmath}
\usepackage{amsthm}
\usepackage{subfigure}
\usepackage{epsfig}
\usepackage{graphics}
\usepackage{dsfont}
\usepackage{color}
\usepackage{verbatim}
\usepackage{latexsym}
\usepackage{booktabs}
\usepackage{multirow}
\usepackage{setspace}
\usepackage{float}
\usepackage{booktabs}
\usepackage{placeins}
\usepackage{bbm}
\usepackage{graphicx}
\usepackage{eurosym}
\usepackage{xcolor}
\usepackage{enumitem}

\newtheorem{theorem}{Theorem}[section]
\newtheorem{proposition}{Proposition}
\newtheorem{lemma}{Lemma}
\newtheorem*{proposition*}{Proposition}
\newtheorem{corollary}{Corollary}
\newtheorem{definition}{Definition}

\newtheorem*{remark}{Remark}
\theoremstyle{definition}
\newtheorem{example}{Example}[section]

\usepackage{setspace}
\usepackage{todonotes}

\newboolean{includeHidden}
\setboolean{includeHidden}{false}
\newcommand{\hide}[1]{\ifthenelse{\boolean{includeHidden}}{{\tiny\textbf{HIDDEN:~}#1}}{}}

\usepackage{rotating} 
\usepackage{natbib}
\bibpunct[, ]{(}{)}{,}{a}{}{,}%
\def\newblock{\ }%
\usepackage{multirow}
\usepackage[caption=false]{subfig}

\usepackage[]{todonotes}
\usepackage{setspace}
\usepackage{xpatch}
\newproof{pf}{Proof}

\newcommand{\q}{\mathbf{q}}
\newcommand{\rbf}{\mathbf{r}}

\newcommand{\Q}{\mathcal Q}

\DeclareMathOperator{\R}{\mathbb{R}}
\DeclareMathOperator{\Z}{\mathbb{Z}}
\DeclareMathOperator{\N}{\mathbb{N}}

\usepackage{amsmath}
\DeclareMathOperator*{\argmax}{arg\,max}

\DeclareMathOperator{\supp}{supp}
\usepackage{tikz}
\usetikzlibrary{arrows}
\usetikzlibrary{calc,fit}
\usetikzlibrary{positioning,chains}

\makeatletter

\def\tmk@labeldef#1,#2\@nil{%
	\def\tmk@label{#1}%
	\def\tmk@def{#2}%
}

\makeatother

\journal{Economic Letters}

\begin{document}
	\singlespacing
	\title{On the Expressiveness of Assignment Messages}
	\author{Maximilian Fichtl, Technical University of Munich}
	\address{Decision Sciences \& Systems (DSS), Department of Informatics (I18), Technische Universität München, Boltzmannstr. 3, 85748 Garching, Germany}
	\ead{max.fichtl@tum.de}
	\begin{abstract}
	In this note we prove that the class of valuation functions representable via integral assignment messages is a proper subset of strong substitutes valuations. Thus, there are strong substitutes valuations not expressible via assignment messages.
	\end{abstract}
	\maketitle
	\section{Introduction}
	Strong substitutes valuations are an important class of valuation functions for indivisible markets, guaranteeing existence of a Walrasian equilibrium. They were introduced by \cite{milgrom2009_ss} as a multi-unit generalization of the gross substitutes condition for single-unit markets \citep{kelso1982_gs}. Many different equivalent definitions of gross- and strong substitutes have been discovered, see for example \cite{leme2017_gs}, \cite{fujishige2003_gs} and \cite{murota2016_dcanalysis}. Intuitively, a buyer with strong substitutes valuation treats different types of goods as one-to-one substitutes, and in particular, there are no complementarities between goods.
	
	In sealed-bid auction settings where bidders may be assumed to have strong substitutes preferences, it is of major practical importance that bidders can efficiently report their preferences to the seller. Explicitly reporting values for every possible bundle is not feasible in practice, since the number of bundles grows exponentially with the number of available goods. Thus, there is the need for a bidding language, allowing bidders to express their strong substitutes preferences in a more compact and intuitive way, while not further restricting the class of expressible valuations. \cite{milgrom2009_am} introduces \emph{integer assignment messages}, in the following only called \emph{assignment messages}, and proves that every valuation function expressible via an assignment message fulfills the strong substitutes condition. While this bidding language is quite intuitive, the question if bidders can express  \emph{arbitrary} strong substitutes valuations with assignment messages remained open. In this note we give a negative answer by proving that there are strong substitutes valuations not expressible via assignment messages. Our proof follows the lines of \cite{ostrovsky2015_endowed}, who showed that a related bidding language, called \emph{endowed assignments}, for single unit markets cannot express arbitrary gross substitutes valuations.
	
	As has recently been shown by \cite{baldwin2021_ss} the Strong Substitutes Product-Mix Auction \citep{klemperer2008_pma, klemperer2010_pma} is capable of expressing arbitrary strong substitutes preferences. Thus, our result implies that the Strong Substitutes Product-Mix Auction remains the only known bidding language that allows bidders to express all such preferences.
	\section{Economic Setting}
	We consider a market with $n \geq 2$ types goods $i\in \{1,\dots,n\}$. A \emph{bundle} of goods is a vector $\mathbf{q} \in \Z^n$, where $q_i$ is the number of units of good $i$ contained in $\mathbf{q}$. A negative value of $q_i$ expresses a willingness to sell goods of type $i$. Bidders' preferences are given by \emph{valuation functions} $v: \Q \rightarrow \R$, where $\Q \subset \Z^n$ is a finite set of feasible bundles with $0 \in \Q$, and $v(\mathbf{q})$ denotes the bidder's value for receiving bundle $\mathbf{q}$. A price vector is a vector $\mathbf{p} \in \R^n$, $p_i$ denoting the cost per unit of good $i$. Given a price vector $\mathbf{p}$, bidders seek to maximize their quasi-linear utility by choosing a bundle from their \emph{demand set}
	\[
	D(\mathbf{p}) = \argmax_{\mathbf{q} \in \Q} v(\mathbf{q})-\langle \mathbf{p}, \mathbf{q} \rangle.
	\]
	The utility of receiving such a bundle is called the \emph{indirect utility} and is denoted by
	\[
	u(\mathbf{p}) = \max_{\mathbf{q} \in \Q} v(\mathbf{q})-\langle \mathbf{p}, \mathbf{q} \rangle.
	\]
	We are interested in markets where bidders' valuations satisfy the strong substitutes condition.
	\begin{definition}[Strong Substitutes \citep{milgrom2009_ss}]
		A valuation function $v: \Q \rightarrow \R$, where $Q \subseteq \Z^n_{\geq 0}$ is \emph{strong substitutes}, if its binary representation $v: \tilde Q \rightarrow \R$, where $\tilde{\mathcal Q} \subseteq \{0,1\}^{\tilde n}$, in which every unit of every good is interpreted as a separate good, satisfies the \emph{binary substitutes} property: for all price vectors $\tilde{\mathbf p}, \tilde{ \mathbf p}' \in \R^{\tilde n}$ with $\tilde{ \mathbf p} \leq \tilde{ \mathbf p}'$ and any $\tilde{ \mathbf x}\in D(\tilde{ \mathbf p})$, there is a bundles $\tilde{ \mathbf x}' \in D(\tilde{ \mathbf p}')$ with $\tilde{ x}'_i \geq \tilde {x}_i$ for all $i$ with $\tilde p'_i = \tilde p_i$.
		
		If $\mathcal Q = \{0,1\}^n$, we call $v$ \emph{gross substitutes}.
	\end{definition}
	Probably the most important feature of strong substitutes valuations is that they ensure the existence of a Walrasian equilibrium \citep{milgrom2009_ss}.
	Assignment messages are a bidding language simplifying the task of expressing a bidder's valuation function $v: \Q \rightarrow \R$.
	\section{Assignment Messages}
	An integer assignment message as introduced in \cite{milgrom2009_am} expresses a bidder's valuation via a linear program. It is determined by a set of $m \in \N$ variables $x_j$ for $j \in J = \{1,\dots,m\}$, where each variable is associated with one of the $n\geq 2$ types of goods $k_j \in \{1,\dots,n\}$, and with a value $v_j \in \R$. We assume that for each good $i$ there is at least one variable associated with it - if not, we can just introduce dummy variables with a value of $0$. We define $R_i = \{j \in J\,:\, k_j = i\}$ to be the set of all variables associated with good $i$. Moreover, the bidder provides a set $\mathcal{I} \subset \mathcal{P}(J)$ of inequalities. Each inequality $I \in \mathcal{I}$ is a subset of the variables $J$ and is associated with an integral upper bound $u(I) \geq 0$ and an integral lower bound $\ell(I) \leq 0$, describing the linear constraints $\ell(I) \leq \sum_{j \in I} x_j \leq u(I)$. The value $v(\mathbf{q})$ for a bundle $\mathbf{q} \in \Q$ is given by
	\begin{align}
		v(\mathbf{q}) = \max& \sum_{j=1}^m v_j x_j \tag{VAL} \label{eq:basic_am} \\
		\text{s.t.} &\, \ell(I) \leq \sum_{j \in I} x_j \leq u(I) \, \forall I \in \mathcal{I} \nonumber \\
		&\,\sum_{j \in R_i} x_j = q_i \, \forall i=1,\dots,n. \nonumber
	\end{align}
	Here, $\Q \subset \Z^n$ is the set of all $\q$ for which (\ref{eq:basic_am}) has a feasible solution, which clearly contains $0$ and is bounded and thus finite. The indirect utility $u(\mathbf{p}) = \max_{\mathbf{q}\in \Q} v(\mathbf{q})-\langle \mathbf{p},\mathbf{q}\rangle$ can be expressed via
	\begin{align}
		u(\mathbf{p}) = \max &\sum_{i=1}^n \sum_{j\in R_i} (v_j - p_i) x_j \tag{IU}  \label{eq:iu} \\
		\text{s.t.} &\, \ell(I) \leq \sum_{j \in I} x_j \leq u(I) \, \forall I \in \mathcal{I}. \nonumber 
	\end{align}
	The demand set $D(\mathbf{p})$ of maximizers of $v(\mathbf{q})-\langle \mathbf{p},\mathbf{q} \rangle$ is the set of all $\mathbf{q} \in \Q$ that can be written as $q_i = \sum_{j \in R_i} x_j$ where $\mathbf{x}$ is an integral solution to (\ref{eq:iu})\footnote{Strictly speaking, we would have to impose integrality constraints on the variables in (\ref{eq:basic_am}) and (\ref{eq:iu}) at this point. However, \cite{milgrom2009_am} shows that both problems always have integral optimal solutions if the constraints have the required structure from Definition \ref{def:am}.}.
	The set of inequalities $\mathcal{I}$ may not be chosen arbitrarily, but must possess a certain tree structure. The following two definitions are taken from \cite{milgrom2009_am}.
	\begin{definition}\label{def:tree}
		A nonempty subset $\mathcal{T} \subseteq \mathcal{P}(J)$ is called a \emph{tree}, if for any  $K,L \subseteq \mathcal{T}$ with $K \cap L \neq \emptyset$ there holds $K \subseteq L$ or $L \subseteq K$. For $K \in \mathcal{T}$, we call the inclusion-minimal set $L \in \mathcal{T}$ with $L \supsetneq K$ the \emph{predecessor} of $K$, if such $L$ exists. Conversely, we call each $K$, such that $L$ is the predecessor of $K$, a \emph{successor} of $L$. We write $s_{\mathcal{T}}(L) = \{K\,:\, L \text{ predecessor of } K \text{ in } \mathcal{T}\}$ for the set of successors of $L$ in $\mathcal{T}$.
	\end{definition}
	\begin{definition}\label{def:am}
		The variables $J$ and inequalities $\mathcal{I}$ define an \emph{assignment message}, if $\mathcal{I} = \mathcal{T}_0 \cup \dots \cup \mathcal{T}_n$ is the union of $n+1$ trees, such that
		\begin{itemize}
			\item for $i =1,\dots,n$, $\mathcal{T}_i$ only contains inequalities in variables associated with good $i$: $\mathcal{T}_i \subseteq \mathcal{P}(R_i)$. Furthermore, $R_i \in \mathcal{T}_i$ and $\{j\} \in \mathcal{T}_i$ for all $j \in R_i$.
			\item $J \in \mathcal{T}_0$ and $\{j\} \in \mathcal{T}_0$ for all $j \in J$. We also write $R_0 = J$.
		\end{itemize}
	Each tree $\mathcal{T}_i$ for $i=0,\dots,n$ contains a unique element $R_i$ without predecessor, which we call the \emph{root} of the tree. The only elements in $\mathcal{T}_i$ that are no predecessors of any other element are the singletons $\{j\}$, which we also call the \emph{terminal nodes}. In the following, we write $s_i(L) := s_{\mathcal{T}_i}(L)$ for the set of successors of $L$ in a specific tree.
	\end{definition}
	Note that since $n \geq 2$, the trees can always be chosen such that they intersect only in the terminal nodes: $\mathcal{T}_0 \cap \mathcal{T}_i = \{\{j\}\,:\, j \in R_i\}$ and $\mathcal{T}_i \cap \mathcal{T}_k = \emptyset$ for $i,k \geq 1$ with $i \neq k$: if $I \in \mathcal{T}_i \cap \mathcal{T}_k\neq \emptyset$ with $i< k$ and $I$ not a singleton, we necessarily have that $i = 0$. Moreover, since we assume that there is at least one variable associated with every good, $R_0 \supsetneq R_j$, so $I$ is neither the root nor a terminal node of $\mathcal{T}_0$. Thus, we can remove $I$ from $\mathcal{T}_0$ without violating Definition \ref{def:am}.
	
	\section{Strong Exchangeability}
	
	\cite{ostrovsky2015_endowed} show that there are gross substitutes valuations that are not expressible via \emph{endowed assignments}. They observe that all endowed assignment valuations satisfy a certain property, called \emph{strong exchangeability}, and provide a gross substitutes valuation that is not strongly exchangeable, which is then consequently not expressible via endowed assignments. For two vectors $\mathbf{p}$ and $\mathbf{q}$, denote by $\supp_+ \mathbf p - \mathbf q$ the set of indices $i$ with $p_i - q_i > 0$.
	\begin{definition}[Single-Unit Strong Exchangeability \citep{ostrovsky2015_endowed}]\label{def:single_se}
		A valuation $v: \{0,1\}^n \rightarrow \R$ satisfies \emph{strong exchangeability}, if for every price vector $\mathbf{p}$ and all bundles $\mathbf{q}, \mathbf{r} \in D(\mathbf{p})$ with
		a minimal number of items, i.e., $\sum_i q_i = \sum_i r_i = \min_{\mathbf{q}' \in D(\mathbf p)} \sum_i q'_i$, there is a bijection $\sigma : \supp_+ \mathbf{q}-\mathbf{r} \rightarrow \supp_+ \mathbf{r}-\mathbf{q}$, such that $\mathbf{q}-\mathbf e_i+\mathbf e_{\sigma(i)}$ and $\mathbf{r}-\mathbf e_{\sigma(i)}+\mathbf e_i$ are contained in $D(\mathbf p)$ for all $i \in \supp_+ \q-\mathbf{r}$. Here, $\mathbf e_i$ is the $i$-th standard unit vector.
	\end{definition}
	\begin{proposition}[\cite{ostrovsky2015_endowed}]\label{prop:gs_se}
		There are gross substitutes valuations not satisfying the strong exchangeability property.
	\end{proposition}
	Our proof follows the same lines: first, we provide a multi-unit extension of strong exchangeability, and then we show that all valuations induced by assignment messages satisfy this property.
	
	\begin{definition}[Multi-Unit Strong Exchangeability]\label{def:se}
		A valuation $v: \Q \rightarrow \R$ satisfies \emph{strong exchangeability}, if for every price vector $\mathbf{p}$ and all bundles $\mathbf{q}, \mathbf{r} \in D(\mathbf{p})$ with minimal number of items, there is a correspondence $\sigma \in \supp_+ \mathbf{q}-\mathbf{r} \times \supp_+ \mathbf{r}-\mathbf{q}$, such that
		\begin{enumerate}[label = \arabic*., ref = \arabic*]
			\item \label{def:ex_in_demand} For each $(i,j) \in \sigma$, $\mathbf{q}-\mathbf{e}_i+\mathbf{e}_j \in D(\mathbf{p})$ and $\mathbf{r}+\mathbf{e}_i-\mathbf{e}_j \in D(\mathbf{p})$
			\item \label{def:ex_existence} For each $i \in \supp_+ \mathbf{q}-\mathbf{r}$ and $j \in \supp_+ \mathbf{r}-\mathbf{q}$, we have $1 \leq |\{j'\,:\, (i,j') \in \sigma \}| \leq q_i-r_i$ and $1 \leq |\{i' \,:\, (i',j) \in \sigma \}| \leq r_j - q_j$.
		\end{enumerate}
	\end{definition}
	\begin{remark}
		Note that in a single-unit market, Property \ref{def:ex_existence} says that for every $i \in \supp_+ \mathbf{q}-\mathbf{r}$, there is exactly one $j \in \supp_+ \mathbf{r}-\mathbf{q}$ such that $(i,j) \in \sigma$ and vice-versa. In this case $\sigma$ can be interpreted as a bijection $\sigma: \supp_+ \mathbf{q}-\mathbf{r} \rightarrow \supp_+ \mathbf{r}-\mathbf{q}$, so for single-unit markets Definitions \ref{def:single_se} and \ref{def:se} are equivalent.
	\end{remark}
	
	In order to prove that every assignment message satisfies strong exchangeability, we show that computing the indirect utility of an assignment message valuation can be interpreted as a min-cost flow problem. Given the tree structure of assignment messages from Definitions \ref{def:tree} and \ref{def:am}, we can transform the indirect utility problem (\ref{eq:iu}) by variable substitution as follows: for each $I \in \mathcal{I}$ introduce a variable $y_I$ representing $y_{I} = \sum_{j \in I} x_j$. Note that since $\{j\} \in \mathcal{I}$ for all $j \in J$, there are variables $y_{\{j\}}$ corresponding to the variables $x_j$. If $I \in \mathcal{T}_i$ is not a singleton, $I$ is the disjoint union of all its successors $K \in \mathcal{T}_i$, so
	\[
	y_{I} = \sum_{j \in I} x_j = \sum_{\substack{K \in s_i(I)}} \sum_{j \in K} x_j = \sum_{K \in s_i(I)} y_K.
	\]
	Similarly, we have
	\[
	y_{R_0} = \sum_{j \in R_0} x_j = \sum_{i=1}^n \sum_{j \in R_i} x_j = \sum_{i=1}^n y_{R_i}.
	\]
	The constraints $\ell(I) \leq \sum_{i \in I} x_i \leq u(I)$ translate to $\ell(I) \leq y_I \leq u(I)$.
	Using these observations, it is not hard to see that Problem (\ref{eq:iu}) can equivalently be formulated as
	\begin{align}
		\min & \sum_{i=1}^n \sum_{j \in R_i}(p_i-v_j) y_{\{j\}} \tag{MCF} \label{eq:mcf} \\
		\text{s.t.} &\, y_I - \sum_{K \in s_0(I)} y_K = 0 \,\forall I \in \mathcal{T}_0\setminus\{\{j\}\,:\, j \in R_0\}  \label{constr:t0} \\
		&\, \sum_{K \in s_i(I)} y_K - y_I = 0 \,\forall I \in \mathcal{T}_i\setminus\{\{j\}\,:\, j \in R_i\} \, \forall i=1,\dots,n \label{constr:ti}\\
		&\, \sum_{i=1}^n y_{R_i} - y_{R_0} = 0 \label{constr:roots}\\
		&\, \ell(I) \leq y_I \leq u(I) \, \forall I \in \mathcal{I} \label{constr:ineq}
	\end{align}
where instead of maximizing the objective function of (\ref{eq:iu}), we minimize the negative objective function to be consistent with literature on min-cost flows. The following lemma is a simple consequence of the variable substitution explained above, so the proof is omitted.
\begin{lemma}
Let $v: \Q \rightarrow \R$ be an assignment message. Then $\mathbf{q} \in D(\mathbf{p})$ if and only if there is an integral solution to (\ref{eq:mcf}) with $q_i = y_{R_i}$ for all $i \geq 1$.
\end{lemma}

It is not hard to see that each variable $y_I$ for $I \in \mathcal{I}$ appears exactly twice in the set of equality constraints of (\ref{eq:mcf}), once with coefficient $1$, and once with coefficient $-1$: for example, if $I \in \mathcal{T}_i$, where $i \geq 1$, is not a singleton, $y_I$ clearly appears with negative sign in one of the equalities in (\ref{constr:ti}). If additionally $I \neq R_i$, $I$ is also the successor of some element, so $y_I$ also appears in exactly one equation from (\ref{constr:ti}) with coefficient $1$. On the other hand, if $I = R_i$, $y_I$ appears with positive coefficient in equation (\ref{constr:roots}). This property can be checked in a similar way for all other variables.
Thus, if we collect the variables $y_I$ in the vector $\mathbf{y} = (y_I)_{I \in \mathcal{I}}$ and write the equality constraints (\ref{constr:t0})-(\ref{constr:roots}) in matrix form as $A\mathbf{y} = 0$, $A$ is the incidence matrix of a directed graph, where $\mathcal{I}$ is the set of arcs, and each of the constraints from (\ref{constr:t0})-(\ref{constr:roots}) corresponds to a vertex in the graph. For any such vertex, $I$ is an ingoing arc, if $y_I$ appears with coefficient $1$, and an outgoing arc, if it appears with coefficient $-1$. Consequently, Problem (\ref{eq:mcf}) can be interpreted as a min-cost flow problem where $y_I$ denotes the flow along arc $I$. For details on min-cost flows, we refer to \cite{ahuja1993_networkflows}.
\begin{example}\label{ex:trees}
	Suppose a bidder submits an assignment message in four variables $J = R_0 = \{1,2,3,4\}$, where $R_1 = \{1,2,3\}$ and $R_2 = \{4\}$. The submitted inequalities induce the trees $\mathcal{T}_0 = \{R_0,\{2,3,4\},\allowbreak \{1\},\dots,\{4\}\}$, $\mathcal{T}_1 = \{R_1,\{1,2\},\{1\},\dots,\{3\}\}$ and $\mathcal{T}_2 = \{R_2\}$. The directed graph corresponding to the incidence matrix $A$ is shown in Figure 1.
\end{example}
\begin{figure}[t]
	\centering
	\includegraphics[height=0.2\textheight]{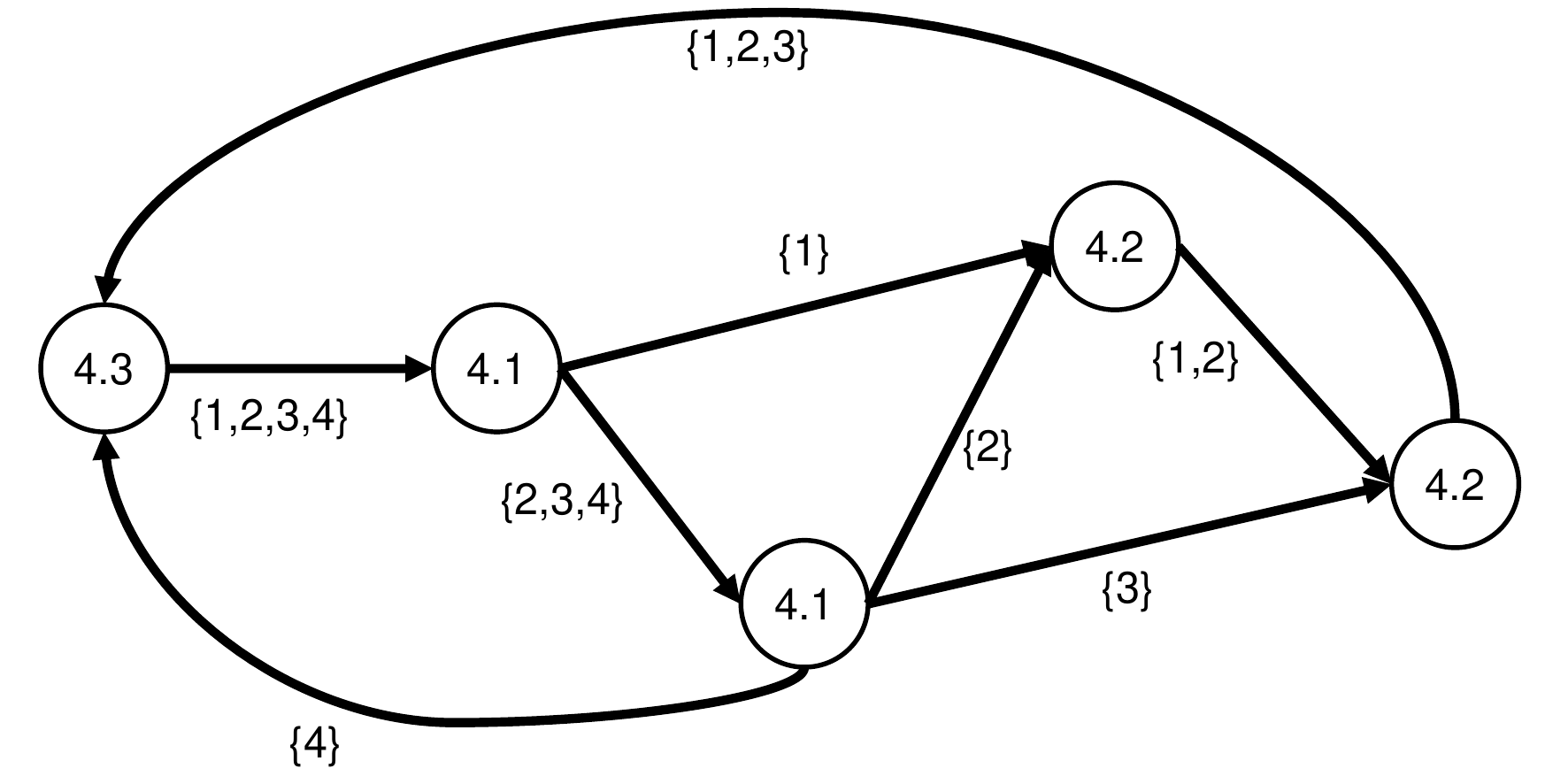}
	\label{fig:graph}
	\caption{Directed graph from Example \ref{ex:trees}. The labels on the vertices indicate the equality constraint in (\ref{eq:mcf}) they correspond to.}
\end{figure}
In order to prove our main Theorem \ref{thm:am_se}, we recall some properties of min-cost flows.
\begin{lemma}[Properties of Flows \citep{ahuja1993_networkflows}\footnote{\cite{ahuja1993_networkflows} generally consider only non-negative flows, as arbitrary flow problems can be easily transformed into non-negative ones by introducing backward-arcs. For the sake of brevity, we allow negative flows here. Note that the proofs given in their book for the mentioned flow properties do actually not require non-negativity.}]\label{lem:flow_props}
	Let $G=(V,A)$ be a directed graph with vertex set $V$ and arc set $A$. Let $\mathbf f:A \rightarrow \Z$ be a flow on $G$.
	\begin{enumerate}[label = \arabic*., ref = \arabic*]
		\item \label{lem:flow_decomp} If $\mathbf f$ is balanced at every vertex, i.e.,
		\[
		\sum_{a=(w,v) \in A} f_a - \sum_{a=(v,w) \in A} f_a = 0
		\]
		for all vertices $v \in V$, then $\mathbf f$ can be decomposed into finitely many cycles: there are subsets $C^1,\dots,C^m \subseteq A$ of arcs, such that each $C^k$ is an undirected cycle in $G$, and balanced flows $\mathbf c^k: C^k \rightarrow \{-1,1\}$, such that $\mathbf f = \sum_{k=1}^m \mathbf c^k$. Moreover, we have $f_a \geq 0 \Leftrightarrow c^k_a \geq 0$ for all $a \in A$ and all $k=1,\dots,m$. \footnote{This follows from the construction in the proof of Theorem 3.5 in \cite{ahuja1993_networkflows}.}
		\item \label{lem:no_neg_cycle} Suppose $\mathbf f$ is an optimal solution to the general min-cost flow problem
		\begin{align*}
			\min & \sum_{a \in A} w_af_a \\
			\text{s.t.} &\, \sum_{a=(w,v) \in A} f_a - \sum_{a=(v,w) \in A} f_a = s(v) \, \forall v \in V\\
			&\, \ell(a) \leq f_a \leq u(a) \, \forall a \in A
		\end{align*}
		for some given weights $w_a$, supplies $s(v)$ and bounds $\ell(a), u(a)$. Then $\mathbf f$ does not contain any negative cycles: for $C \subseteq A$ an undirected cycle and a balanced flow $\mathbf c: C \rightarrow \Z$ on $C$ such that $\mathbf f+\mathbf c$ is a feasible solution, we have that $\sum_{a \in A} w_ac_a \geq 0$. \footnote{This is Theorem 3.8 in \cite{ahuja1993_networkflows}. Since we allow negative flows, we do not need to introduce the residual graph of a flow problem.}
	\end{enumerate}
\end{lemma}
\begin{theorem}\label{thm:am_se}
	Let $v$ be a valuation induced by an assignment message. Then $v$ satisfies the strong exchangeability property.
\end{theorem}
\begin{proof}
Let $\mathbf{q}, \mathbf{r} \in D(\mathbf{p})$ be bundles containing a minimal number of goods and $\mathbf{y}^{\mathbf{q}}, \mathbf{y}^{\mathbf{r}}$ corresponding integral solutions to (\ref{eq:mcf}) with $q_i = y^{\mathbf{q}}_{R_i}$ and $r_i = y^{\mathbf{r}}_{R_i}$ for all $i$. We are going to construct a correspondence $\sigma$ satisfying the properties from Definition \ref{def:se}. Since $\mathbf{y}^{\mathbf{q}}$ and $\mathbf{y}^{\rbf}$ are balanced, i.e., $A \mathbf y^{\mathbf q} = A \mathbf y ^{\mathbf r} = 0$, so is $\mathbf{y}^{\mathbf{q}}-\mathbf{y}^{\mathbf{r}}$, and we can write $\mathbf{y}^{\mathbf{q}} - \mathbf{y}^{\mathbf{r}} = \mathbf{c}_1 + \dots + \mathbf{c}_m$ where each $\mathbf{c}_k$ is supported on a cycle $C^k$ by Lemma \ref{lem:flow_props}. We prove that the flows $\mathbf{c}^k$ have the following properties:
\begin{enumerate}[label = (\roman*), ref = (\roman*)]
	\item \label{proof:opt_sln} $\mathbf{y}^{\mathbf{r}}+\mathbf{c}^k$ and $\mathbf{y}^{\mathbf{q}}-\mathbf{c}^k$ are optimal solutions to (\ref{eq:mcf}) for every $k$. 
	\item $|\{i\geq 1\,:\, c^k_{R_i} = 1 \}| = |\{i\geq 1\,:\, c^k_{R_i} = -1\}| \leq 1$ for every $k$. \label{proof:sup_goods}
\end{enumerate}
To see \ref{proof:opt_sln}, we first note that $\mathbf{y}^{\mathbf{r}}+\mathbf{c}^k$ is feasible for Problem (\ref{eq:mcf}): as $\mathbf y^{\mathbf r}$ and $\mathbf c^k$ are balanced, so is $\mathbf y^{\mathbf r}+\mathbf c^k$. Concerning the inequality constraints (\ref{constr:ineq}), if $c^k_I > 0$, then it follows from Property \ref{lem:flow_decomp} in Lemma \ref{lem:flow_props} that $u(I) \geq y^{\mathbf{q}}_I \geq y^{\mathbf{r}}_I + c^k_I > \ell(I)$. With a similar argument we can treat the case $c^k_I < 0$. Consequently, by Property \ref{lem:no_neg_cycle}, we have 
\[
\sum_{i=1}^n\sum_{j \in R_i} (p_i-v_j)c^k_{\{j\}} \geq 0.
\]
With the same argument applied to $\mathbf{y}^{\mathbf{q}}-\mathbf{c}^k$, we get
\[
\sum_{i=1}^n\sum_{j \in R_i} (p_i-v_j)c^k_{\{j\}} \leq 0,
\]
so $\sum_{i=1}^n\sum_{j \in R_i} (p_i-v_j)c^k_{\{j\}} = 0$. Hence, the objective values in (\ref{eq:mcf}) of the flows $\mathbf{y}^{\mathbf{q}}$, $\mathbf{y}^{\mathbf{r}}$, $\mathbf{y}^{\mathbf{r}}+\mathbf{c}^k$ and $\mathbf{y}^{\mathbf{q}}-\mathbf{c}^k$ are all equal and thus optimal.

Let us now prove \ref{proof:sup_goods}. To that goal, note that, since $\mathbf{q}$ and $\mathbf{r}$ are bundles with a minimum number of elements, we have $\mathbf{y}^{\mathbf{q}}_{R_0} = \mathbf{y}^{\mathbf{r}}_{R_0}$, so by Property \ref{lem:flow_decomp} of Lemma \ref{lem:flow_props} we have $c^k_{R_0} = 0$ for all $k$. Consider the flow of $\mathbf c^k$ through the vertex corresponding to constraint (\ref{constr:roots}), i.e., representing the equality
\[
\sum_{i=1}^n c^k_{R_i} - c^k_{R_0} = 0.
\]
As $\mathbf c^k$ is supported on a cycle, at most two of the appearing variables $c^k_{R_i}$ can be nonzero. Thus, since $c^k_{R_0} = 0$, either no or exactly two of the $c_{R_i}$ are nonzero, and since their sum equals $0$, one must be $1$, and the other must be $-1$.

Let us now define $\sigma \in \supp_+ \mathbf{q}-\mathbf{r} \times \supp_+ \mathbf{r}-\mathbf{q}$ by
\[
\sigma = \left\{ (i,j)\,:\, \exists k:\, c^k_{R_i} = 1 \wedge c^k_{R_j} = -1 \right\}.
\] 
We check that $\sigma$ has the required properties from Definition \ref{def:se}: let $(i,j) \in \sigma$. Then there is some $\mathbf{c}^k$, such that $c^k_{R_i} = 1$, $c^k_{R_j} = -1$ and $c^k_{R_l} = 0$ for $l \not\in \{i,j\}$. From observation \ref{proof:opt_sln} above we have that $\mathbf{y}^{\mathbf{r}} + \mathbf{c}^k$ is an optimal solution to problem (\ref{eq:mcf}), and the demanded bundle corresponding to that solution is $\mathbf{r}+\mathbf{e}_i-\mathbf{e}_j$. Similarly, the requested bundle corresponding to $\mathbf y^{\mathbf q} - \mathbf c^k$ is $\mathbf{y}^{\mathbf{q}} - \mathbf{e}_i+\mathbf{e}_j$, so Property
\ref{def:ex_in_demand} from Definition \ref{def:se} is satisfied.

For Property \ref{def:ex_existence} from Definition \ref{def:se}, let $i \in \supp_+ \q-\rbf$. We need to show that $1 \leq |\{j'\,:\, (i,j') \in \sigma\}| \leq q_i-r_i$. Since $q_i > r_i$ and $\mathbf{y}^{\mathbf{q}}-\mathbf{y}^{\mathbf{r}} = \mathbf{c}^1 + \dots + \mathbf{c}^m$, there must be some $k$ with $c^k_{R_i} = 1$. Consequently, there is some $j$ with $c^k_{R_j} = -1$, which proves the lower bound. Moreover, by Property \ref{lem:flow_decomp} of Lemma \ref{lem:flow_props}, there is no flow $\mathbf{c}^k$ with $c^k_{R_i} = -1$. Thus, there are at most $q_i-r_i$ flows with $c^k_{R_i} = 1$, proving the upper bound.
\end{proof}
Theorem \ref{thm:am_se} together with Proposition \ref{prop:gs_se} directly imply that assignment messages do not cover all strong substitutes valuations.
\begin{corollary}
	There are strong substitutes valuations that are not representable via an Assignment Messages.
\end{corollary}
\begin{proof}
	Each assignment message satisfies the strong exchangeability property from Definition \ref{def:se}. However, by Proposition \ref{prop:gs_se} by \cite{ostrovsky2015_endowed}, there exist gross substitutes valuations that are not strongly exchangeable. Since gross substitutes valuations are a subset of strong substitutes valuations, and Definitions \ref{def:single_se} and \ref{def:se} are equivalent for single-unit markets, the result follows.
\end{proof}

\section*{Acknowledgements}
I would like to thank Edwin Lock from the University of Oxford for his very valuable comments and suggestions.


\begin{thebibliography}{}

\bibitem[\protect\astroncite{Ahuja et~al.}{1993}]{ahuja1993_networkflows}
Ahuja, R.~K., T.~L. Magnanti, and J.~B. Orlin\leavevmode\nopagebreak\newline
  1993.
\newblock {\em Network Flows - Theory, Algorithms, and Applications}.
\newblock London: Prentice Hall.

\bibitem[\protect\astroncite{Baldwin and Klemperer}{2021}]{baldwin2021_ss}
Baldwin, E. and P.~Klemperer\leavevmode\nopagebreak\newline 2021.
\newblock Proof that the strong substitutes product-mix auction bidding
  language can represent any strong substitutes preferences.
\newblock \url{https://elizabeth-baldwin.me.uk/papers/strongsubsproof.pdf}.

\bibitem[\protect\astroncite{Fujishige and Yang}{2003}]{fujishige2003_gs}
Fujishige, S. and Z.~Yang\leavevmode\nopagebreak\newline 2003.
\newblock A note on kelso and crawford's gross substitutes condition.
\newblock {\em Mathematics of Operations Research}, 28(3):463--469.

\bibitem[\protect\astroncite{Kelso and Crawford}{1982}]{kelso1982_gs}
Kelso, A.~S. and V.~P. Crawford\leavevmode\nopagebreak\newline 1982.
\newblock Job matching, coalition formation, and gross substitutes.
\newblock {\em Econometrica}, 50(6):1483--1504.

\bibitem[\protect\astroncite{Klemperer}{2008}]{klemperer2008_pma}
Klemperer, P.\leavevmode\nopagebreak\newline 2008.
\newblock A new auction for substitutes: Central bank liquidity auctions, the
  u.s. tarp, and variable product-mix auctions.
\newblock
  \url{https://www.nuffield.ox.ac.uk/economics/Papers/2008/substsauc.pdf}.

\bibitem[\protect\astroncite{Klemperer}{2010}]{klemperer2010_pma}
Klemperer, P.\leavevmode\nopagebreak\newline 2010.
\newblock The product-mix auction: a new auction design for differentiated
  goods.
\newblock {\em Journal of the European Economic Association}, 8:526--36.

\bibitem[\protect\astroncite{Milgrom}{2009}]{milgrom2009_am}
Milgrom, P.\leavevmode\nopagebreak\newline 2009.
\newblock Assignment messages and exchanges.
\newblock {\em American Economic Journal: Microeconomics}, 1(2):95--113.

\bibitem[\protect\astroncite{Milgrom and Strulovici}{2009}]{milgrom2009_ss}
Milgrom, P. and B.~Strulovici\leavevmode\nopagebreak\newline 2009.
\newblock Substitute goods, auctions, and equilibrium.
\newblock {\em Journal of Economic Theory}, 144(1):212--247.

\bibitem[\protect\astroncite{Murota}{2016}]{murota2016_dcanalysis}
Murota, K.\leavevmode\nopagebreak\newline 2016.
\newblock Discrete convex analysis: A tool for economics and game theory.
\newblock {\em The Journal of Mechanism and Institution Design}, 1(1):151--273.

\bibitem[\protect\astroncite{Ostrovsky and
  Paes~Leme}{2015}]{ostrovsky2015_endowed}
Ostrovsky, M. and R.~Paes~Leme\leavevmode\nopagebreak\newline 2015.
\newblock Gross substitutes and endowed assignment valuations.
\newblock {\em Theoretical Economics}, 10(3):853--865.

\bibitem[\protect\astroncite{Paes~Leme}{2017}]{leme2017_gs}
Paes~Leme, R.\leavevmode\nopagebreak\newline 2017.
\newblock Gross substitutability: An algorithmic survey.
\newblock {\em Games and Economic Behavior}, 106(C):294--316.

\end{thebibliography}

\end{document}